\documentclass[a4paper,journal]{IEEEtran}
\IEEEoverridecommandlockouts
\usepackage[backend=bibtex, style=ieee, sorting=none]{biblatex}
\usepackage{amsmath,amssymb,amsthm}
\usepackage{graphicx}
\usepackage{orcidlink}
\usepackage[ruled,vlined,linesnumbered]{algorithm2e}
\newtheorem{definition}{Definition}
\newtheorem{assumption}{Assumption}
\newtheorem{lemma}{Lemma}
\newtheorem{proposition}{Proposition}
\newtheorem{theorem}{Theorem}

\begin{document}
	\title{Agentic Fog:
		A Policy-driven Framework for Distributed Intelligence in Fog Computing}
	\author{Saeed Akbar \orcidlink{0000-0002-7093-9318}, Muhammad Waqas \orcidlink{0000-0001-6489-2819}, and Rahmat Ullah \orcidlink{0000-0001-5162-5164} 
		
		\thanks{Saeed Akbar is with the Department of Computer Science \& Information Technology, Sarhad University of Science \& Information Technology, PK (email: saedakbr@suit.edu.pk).}
		\thanks{Muhammad Waqas is with the Faculty of Electrical Engineering and Computer Science, GIK Institute of Engineering Sciences \& Technology, Topi, Swabi, PK.}    
		\thanks{Rahmat Ullah is with the School of Computer Science \& Electronic Engineering, University of Essex, United Kingdom.}
	}

	\IEEEtitleabstractindextext{
		\begin{abstract}
			Fog and edge computing require adaptive control schemes that can handle partial observability, severe latency requirements, and dynamically changing workloads. Recent research on Agentic AI (AAI) increasingly integrates reasoning systems powered by Large Language Models; however, these tools are not applicable to infrastructure-level systems due to their high computational cost, stochastic nature, and poor formal analyzability. In this paper, a generic model, Agentic Fog (AF), is presented, in which fog nodes are represented as policy-driven autonomous agents that communicate via p2p interactions based on shared memory and localized coordination. The suggested architecture decomposes a system's goals into abstract policy guidance and formalizes decentralized fog coordination as an exact potential game. The framework is guaranteed to converge and remain stable under asynchronous updates, bounded-rational best-response dynamics, and node failures. Simulations demonstrate that the AF system achieves lower average latency and adapts more efficiently to varying demand than greedy heuristics and integer linear programming under dynamic conditions. The sensitivity analysis also demonstrates the capability to perform optimally under different memory and coordination conditions. 
		\end{abstract}

	}
	\maketitle
	\IEEEdisplaynontitleabstractindextext
	
	\section{Introduction}
	
	Fog computing extends cloud computing by moving cloud-based services to the network edge, thereby reducing latency, increasing bandwidth usage, and supporting context-aware computation \cite{alsharif2025survey}. The proximity of storage, computation and control infrastructure to end users offered by fog architecture can serve latency-sensitive and data-intensive applications that are underserved by centralized cloud paradigms \cite{vinueza2024fog, yadav2024efficient}. Fog deployment based on a mesh topology facilitates peer-to-peer (p2p) coordination among Fog Nodes (FNs), thereby reducing dependence on centralized controllers. This makes these systems resilient to failures and may contribute to efficiency gains \cite{wu2024topology}.
	
	Despite the aforementioned developments, most modern fog computing systems are still essentially reactive. Resource allocation, routing of requests, and load balancing are largely controlled by either heuristic approaches, local greedy optimization schemes, or centrally coordinated optimization models that are recalculated periodically based on the assumption of global observability \cite{khaledian2024ai, taleb2025survey}. These techniques struggle to cope with the inherent uncertainty, limited visibility, and stochastic nature of the real-world fog environments \cite{jin2025dynamic}.
	
	These restrictions appear in three critical dimensions. First, the implicit assumptions of centralized and quasi-centralized control structures are near-complete knowledge of system state, which is unsustainable in large-scale, heterogeneous, and failure-prone fog deployments \cite{taleb2025survey}. Second, optimization designs for decentralized fog \cite{mann2022decentralized} do not focus on long-term behavior, making them ineffective under shifting demand patterns \cite{taleb2025survey}. Third, Integer Linear Programming (ILP) and related formulations exhibit poor scalability and result in high computational overhead under dynamic workloads \cite{lera2024multi}.
	
	System-level intelligence emerges from coordinated agent actions \cite{dwivedi2025agentic}. 
	Early agent-based systems were reactive and had limited scope as they utilized predefined rules and static knowledge-bases to operate \cite{sapkota2025ai}. Recent studies have focused on making agents adaptable using statistical learning and Reinforcement Learning (RL). The main limitation of these models is that they are confined to a single task. To tackle the said issue, large foundation models emerged that enable complex workflows incorporating contextual reasoning and external tool usage. Although AI agents are capable of performing multi-step execution, they lack persistent goals and have limited autonomy \cite{schneider2025generative}. AAI is an emerging paradigm where autonomous agents having partial knowledge of the system interact with each other using shared memory and policy alignment to attain distributed intelligence. AAI allows the AI agents to learn from experience and refine their behavior through iterative policy updates, leading to emergent intelligence beyond isolated tasks \cite{bandi2025rise, murugesan2025rise}.
	
	Recent investigations focus primarily on LLM-driven agents \cite{zheng2026agentvne, dell2025agentic}; however, LLM-based agentic systems are not appropriate for  infrastructure-level systems, such as fog computing \cite{abdelzaher2025bottlenecks} due to their computational overhead, and stochastic decision-making behavior. Computation efficiency, predictability, and formal analyzability are the main requirements for such systems. Hence, there is a need for non-LLM-driven AAI framework to cope with the said challenges.
	
	This work formalizes a mesh-fog architecture as a natural instantiation of AAI where Fog Nodes (FNs) exhibit local autonomy, partial observability, p2p interaction, and extended operating time. Moreover, the proposed system provides a theoretical model of coordination, convergence, and stability of the system under partial observability. In other words, this paper presents a framework called Agentic Fog (AF) where FNs act as autonomous agents coordinating through shared memory and p2p localized interactions. In order to achieve decentralized optimization under bounded rationality, the proposed system decomposes the system level objective into policy guidance. In addition, the resulting coordination among the agents follows a Potential Game (PG) formulation ensuring formal convergence guarantees and stability under asynchronous updates and partial failures.
	
	\subsection{Contributions and Novelty of the Proposed Work}
	
	This work makes conceptual and technical contributions to distributed intelligence in fog computing. To clearly distinguish novelty from prior work in multi-agent systems, distributed control, and fog optimization, we summarize the primary contributions as follows:
	
	\begin{itemize}
		\item This study provides a precise, system-centric definition of AAI that is explicitly independent of LLMs. This interpretation emphasizes persistent autonomy, policy-driven decision-making, and analyzable coordination dynamics. This also addresses a growing conceptual ambiguity in the literature and repositions AAI as a viable paradigm for resource-constrained realistic infrastructures.
		
		\item This work models fog computing as a collection of autonomous and policy-driven agents. While agent-based fog control has been explored in the literature, this work provides a unified formal model that explicitly captures shared memory, decentralized coordination, and time-scale separation between slow policy alignment and fast local decision-making.
		
		\item We show that decentralized fog coordination induces an exact PG in which each agent optimizes its expected marginal contribution to a global objective. This formulation yields formal guarantees of convergence under asynchronous bounded-rational best-response dynamics.
		
		\item This study establishes that the proposed AF system remains stable under permanent node failures, provided connectivity and shared memory persistence are maintained. This result formalizes graceful degradation properties that are often observed empirically but rarely proven in practical frameworks.
		
	\end{itemize}
	
	Collectively, the proposed framework is a principled, analyzable, and infrastructure-compatible agentic architecture for fog computing. It complements—but is fundamentally different from both classical distributed control and LLM-centric agent systems.
	
	\subsection{Comparison with Existing Fog Control Paradigms}
	
	Classical ILP-based and RL-based systems, as well as Multi-Agent Systems (MASs), share certain surface-level similarities such as decentralized decision-making and agent-based abstractions with the proposed system. However, there is a fundamental difference between the proposed system and the traditional system in terms of observability, coordination, formal analyzability, and temporal optimization.
	
	ILP-based systems work on global system snapshots while assuming strong observability \cite{yu2025hybrid}. According to literature, these systems are not scalable and fragile under dynamic workloads \cite{lera2024multi}. Compared to ILP-based models, RL–based systems have proved to be adaptive; however, they do not provide formal convergence guarantees and may face stability issues under partial observability  \cite{liu2022partially}. 
	In contrast, the proposed AF is a formally analyzable agentic framework for fog environment that combines persistent autonomy, shared memory, and policy-driven coordination among AI agents.

	\section{Agentic AI Without Large Language Models}

	AAI refers to a computational paradigm in which autonomous entities (agents) maintain an internal state of the system by partially observing the system environment. 
	Formally, an agent \( \mathcal{A}_i \) is mathematically modeled as a tuple as shown in Eq. (\ref{eq:agent}):
	\begin{equation}
		\mathcal{A}_i = \langle \mathcal{O}_i, \zeta_i, \pi_i, \Lambda_i, \upsilon_i \rangle
		\label{eq:agent}
	\end{equation}
	where \( \mathcal{O}_i \) denotes an observation function mapping environmental signals to agent-perceived information, \( \zeta_i \) represents the agent’s internal state or memory, capturing historical context and local knowledge, \( \pi_i : (\mathcal{O}_i, \zeta_i) \rightarrow \Lambda_i \) is a decision policy that maps observations and internal state to actions, \( \Lambda_i \) is the agent’s action space, \( \upsilon_i \) is an optional reward or utility function guiding adaptation, learning, or policy evaluation.

	A defining characteristic of AAI is that the policy function \( \pi_i\) is algorithm-agnostic. It may be instantiated using optimization-based control, game-theoretic strategies, or learning-based methods. Crucially, this formulation does not require natural language reasoning or LLMs. It separates agentic intelligence from language-centric architectures and enables the formal analysis of system behavior.

	Recent literature often describes LLM-based systems as agentic architectures \cite{sapkota2025ai, kostopoulos2025agentic,  acharya2025agentic, brohi2025research}. In such systems, language models perform planning, task decomposition, or tool invocation via probabilistic text generation. Although these system have proven to be very good at general-purpose reasoning, they do not regard system constraints such as timing, resource, or stability. 
	In contrast, classical AAI emphasizes explicit policies, bounded decision spaces, and predictable execution semantics. These properties enable formal reasoning about convergence, stability, and performance guarantees. These requirements are particularly critical in latency-sensitive, resource-constrained environments, such as fog and edge computing infrastructures. As a result, agentic behavior should be understood as arising from sustained autonomous interaction and policy-driven decision-making, rather than from language-based reasoning alone.
	
	The foundation of AAI was laid several decades before the emergence of large language models. Foundational research in MAS, reactive and deliberative agent architectures, swarm intelligence, game theory, and RL established the core principles of autonomous perception, localized decision-making, and distributed coordination \cite{Wooldridge2009MAS,Shoham2009MAS, li2022applications, ma2024efficient}. These early agentic systems demonstrated that coherent global behavior can emerge from local interaction. Specifically, where agents interact under partial observability and limited communication. 
	
	In particular, work on distributed control and game-theoretic learning showed that stability and convergence properties can be achieved through structured local interactions and bounded-rational decision rules \cite{Fudenberg1998Learning,Monderer1996Potential}. RL further extended agentic characteristics by enabling agents to adapt their policies based on environmental feedback, albeit often with limited guarantees in non-stationary multi-agent settings \cite{Sutton2018RL}. 
	The literature on networking and distributed systems have widely used agent-based approaches to solve a variety of optimization problems. Common applications include adaptive routing, load balancing, cache placement, fault recovery, load and resource scheduling \cite{alsharif2025survey,li2022applications}. 
	
In fog and edge computing environments, agent-based systems typically operate under partial observability using localized measurements. Consequently, the underlying algorithmic policies are designed to satisfy strict performance constraints. These metrics include bounded latency, scalability, and failure resilience \cite{li2022applications, Mao2022JSAC}.

More recent literature increasingly aligns AAI with LLM-based systems. In such frameworks, LLMs perform reasoning, planning, and coordination through probabilistic text generation and tool invocation \cite{dwivedi2025agentic, sapkota2025ai}. These approaches emphasize expressive, general-purpose reasoning and have demonstrated impressive performance in open-ended tasks and human-in-the-loop workflows.

LLM-based agent architectures introduce non-determinism, limited interpretability, and substantial computational overhead. Such limitations make them difficult to formally analyze or deploy in resource-constrained infrastructures \cite{murugesan2025rise}. As a result, AAI has become conceptually overloaded, often confusing persistent autonomy and policy-driven decision-making with LLM-based reasoning. This overlooks the optimal alternative agentic systems that are based on traditional MASs and distributed control which are more in line with the operational requirements of fog and edge systems. 
This study provides a precise conceptual clarification of LLM-independent AAI, introducing a mesh-based fog architecture rooted in traditional agent autonomy.

\section{System Model}

To formalize the system, we model the fog infrastructure as an undirected graph $G$, mathematically represented as:
\begin{equation}
	G = (V,E)
	\label{eq:undirectedgraph}
\end{equation}
where each node $v_i \in V$ represents a fog node with bounded computation, storage, and communication capacity, and each edge $E_{ij} \in E$ denotes a bidirectional link between nodes $v_i$ and $v_j$. Each node $v_i$ hosts an autonomous fog agent $A_i$ defined as:
\begin{equation}
	A_i = (\zeta_{i}, \mathcal{M}_i, \pi_i)
	\label{eq:autonomousfogagent}
\end{equation}
where $\zeta_{i}$ is the locally observable system state, including queue lengths, cache occupancy, and link conditions, $\mathcal{M}_i$ is the agent’s finite local memory storing historical observations and outcomes, and $\pi_i : (\zeta_{i}, \mathcal{M}_i, \mathcal{M}^{*}) \rightarrow \Lambda_i$ is a policy mapping local state and shared context to actions.

Agents access a shared agentic memory $\mathcal{M}^{*}$, as defined in Eq. (\ref{eq:sharedmemory}). Shared memory is assumed to be eventually consistent and persists across agent failures.
\begin{equation}
	\mathcal{M}^{*} = \{\mathcal{T}, \mathcal{D}, \mathcal{H}\}
	\label{eq:sharedmemory}
\end{equation}
where $\mathcal{T}$ stores high-level topology information, $\mathcal{D}$ maintains demand history, and $\mathcal{H}$ records summarized historical policy outcomes.

The system seeks to minimize the global objective as shown in Eq. (\ref{eq:globalobjective}):

\begin{equation}
	\mathcal{O}^{*} = \alpha L + \beta C + \gamma R,
	\label{eq:globalobjective}
\end{equation}

where $L$, $C$ and $R$ denote end-to-end latency, aggregate resource consumption and reliability risk respectively. The weights $\alpha, \beta, \gamma \geq 0$ encode system-level priorities.

\begin{definition}[Agentic Fog System]
	A fog system is agentic if global behavior (intelligence) emerges from the coordinated decisions of multiple autonomous agents operating under partial observability and shared memory, without relying on centralized real-time optimization.
\end{definition}

\begin{figure*}[htp]
	\centering
	\includegraphics[width=0.8\linewidth]{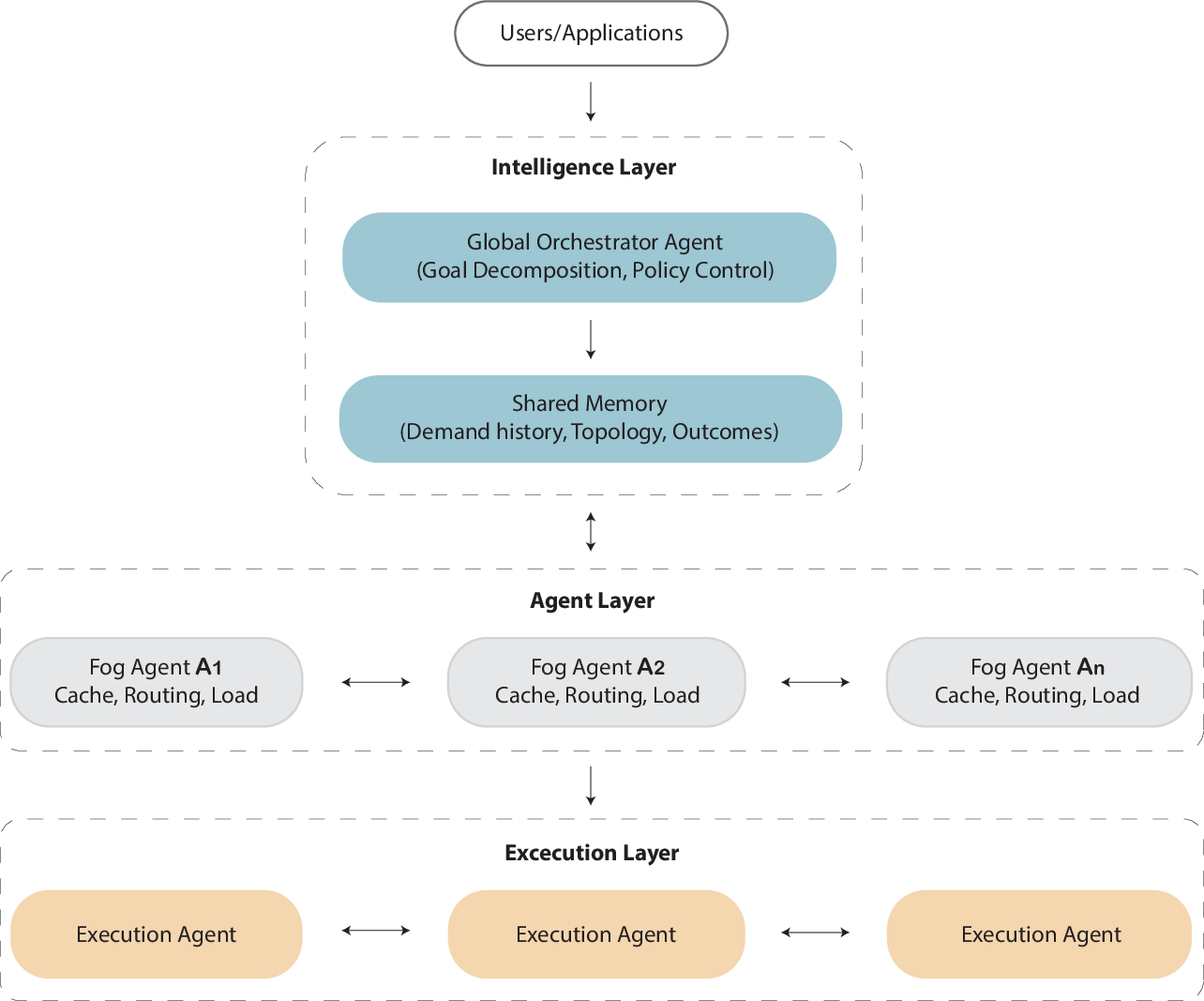}
	\caption{Layered architecture of the proposed Agentic Fog system comprising the Intelligence Layer, the Agent Layer, and the Execution Layer. The Intelligence layer has a Global Orchestrator Agent for policy control and goal decomposition, and a Shared Memory for saving previous system states, the Agent Layer consists of distributed Fog Agents managing caching, routing, and load balancing, and  finally, the Execution Layer comprises of Execution Agents responsible for localized task execution.}
	\label{fig:architecture}
\end{figure*}

Fig. \ref{fig:architecture} illustrates the proposed architecture designed to realize AAI at the system level. Intelligence emerges from coordinated interactions among autonomous FAs rather than from centralized optimization or monolithic controllers. Each architectural component is motivated by specific theoretical and practical considerations to ensure scalability, adaptability, and robustness in non-stationary, partially observable fog environments.

\subsection{Separation of System-Level Cognition and Local Autonomy}

The proposed architecture decouples system-level cognition from local decision-making. The Global Orchestrator Agent (GOA) is responsible for the decomposition and interpretation of strategic goals into abstract sub-objectives. Based on local observations, each FA constructs its own solution to attain system goals while avoiding scalability issues. In addition, the separation among different agents ensures global coherence through policy alignment. Theoretically, the division of global objective into locally optimizable sub-objectives for the PG formulation ensures convergence guarantees in our proposed work.

\subsection{Shared Memory}

The proposed architecture uses shared memory as a coordination tool among the participating agents in our proposed PG. It supports non-Markovian coordination and avoids centralized authority for decision-making. The participating agents share abstracted knowledge instead of raw state information. The agents can reason and find temporal patterns (demand trends and prior outcomes) from the abstracted knowledge while keeping the communication overhead as low as possible. Theoretically, the proposed system aligns local utilities with the global objective using the shared memory. This ensures that each agents contributes to the global optimization by unilateral improvements through local utility optimization; this ultimately leads the system to stable equilibrium states after finite number of interactions.

\subsection{Fog Agents}

FAs are designed as autonomous agents with partial observability, bounded rationality, and adaptive policies. The participating agents independently estimate their marginal contribution to the global objective using their local observations and past history from shared memory enabling quick and optimal response to local dynamics, such as workload changes or link failures. At the same time, p2p coordination among FAs enables cooperative behaviors, such as cache coordination and load redistribution, which are essential for ensuring system-level intelligence. In addition, the proposed system facilitates balance between agents autonomy and coordination to improve scalability and resilience of the underlying architecture.

\subsection{Peer-to-Peer Coordination for Emergent Intelligence}

The localized coordination without global synchronization between FAs is facilitated by direct p2p communication. This allows quick conflict resolution and consistency upholding in overlapping service areas and not subject to the overhead of centralized communication. Such interactions, by theory, provide the strategic coupling required to make the AF architecture modeled a PG, meaning that it converges to an equilibrium with asynchronous updates of best-responses.

\subsection{Decoupling Decision-Making from Execution}

The architecture separates decision-making and execution with the introduction of lightweight Execution Agents (EAs) that execute specific tasks of concrete operations. Where FAs are only interested in strategic thought and coordination, EAs carry out stateless services like content serving and replication and failover. Such isolation promotes modularity and enables the execution components to be scaled without any effect on the decision logic. It also makes sure that the failures on the execution-level do not impact the stability of the agentic coordination.

\subsection{Layered Structure for Scalability and Robustness}

A strategic decision to deal with the complexity of the proposed system and increase the robustness is the rational division into Agentic Intelligence, Fog Agent and Execution layers. All the layers are of different temporal and spatial scales, which provides local and global adaptability to the system. This makes our proposed system gracefully degradable in the event of partial failures in a heterogeneous Fog computing environment.

\subsection{Design Implications}

Collectively, the above design decisions make sure that the proposed fog computing system is a cognitive distributed environment instead of an optimized pipeline. It provides direct support to the theoretical results regarding the convergence, stability and adaptability. Intelligence is the result of interaction, memory and coordination. Thus, the architecture fills the gap between the theory of AAI and the practical fog systems.

Importantly, the GOA operates on a slower timescale and provides policy guidance rather than explicit control actions. This decoupling ensures that convergence and stability guarantees do not depend on the continuous availability of the GOA.

\begin{algorithm}[t]
	\caption{Global Objective Decomposition and Policy Alignment}
	\label{alg:orchestrator} \vspace{0.35em}
	\KwIn{Global objective $\mathcal{O}^{*}$, system constraints $\mathcal{C}$}
	\KwOut{Policy guidance $\Pi_G$ stored in shared memory}\vspace{0.75em}
	
	Monitor long-term demand statistics and aggregated system metrics\
	Decompose $\mathcal{O}^{*}$ into abstract sub-objectives $\{\mathcal{O}^{*}_k\}$ in terms of latency, cost, and reliability\;\vspace{0.35em}
	Obtain policy preferences and soft constraints $\Pi_G$ (e.g. weight bounds, priority signals)\;\vspace{0.35em}
	Publish $(\Pi_G, \mathcal{C})$ to $\mathcal{M}^{*}$\; \vspace{0.35em}
\end{algorithm}

\begin{algorithm}[t]
	\caption{Fog Agent Policy Update}
	\label{alg:fog_agent} \vspace{0.35em}
	\KwIn{$\zeta_{i}$, $\mathcal{M}_i$, $\mathcal{M}^{*}$}
	\KwOut{Updated $\pi_i$}\vspace{0.75em}
	
	Observe current $\zeta_{i}$\;\vspace{0.35em}
	Get $(\Pi_G, \mathcal{C})$ out of $\mathcal{M}^{*}$\;\vspace{0.35em}
	
	\ForEach{action $a \in \Lambda_i$}{
		Calculate the expected marginal contribution using:
		\[
		\upsilon_i(a) \triangleq \mathbb{E}[\Delta \mathcal{O}_i \mid \zeta_{i}, \mathcal{M}_i, \mathcal{M}^{*}, a]
		\]
	}\vspace{0.35em}
	
	Choose best response based on:
	\[
	a_i^\star \gets \arg\min_{a \in \Lambda_i} \upsilon_i(a)
	\]
	
	Delegate $a_i^\star$ to EA\;\vspace{0.35em}
	Observe result of execution and local effect\;\vspace{0.35em}
	Update $\mathcal{M}_i$\;\vspace{0.35em}
	Store outcomes summary to $\mathcal{M}^{*}$\;\vspace{0.35em}
\end{algorithm}

Algorithm \ref{alg:orchestrator} captures the system-level intelligence in the proposed AF system in which the GOA is intentionally excluded from real-time control. Its primary function is to translate the global objective $\mathcal{O}^{*}$ into abstract policy guidance. The GOA stores the alignment signals in shared memory to ensure consistency across FAs without sacrificing decentralization. This design promotes scalability and provides the structural basis necessary for modeling the system as a PG.

\begin{algorithm}[t]
	\caption{Peer-to-Peer Fog Agent Coordination}
	\label{alg:coordination} \vspace{0.35em}
	\KwIn{Neighbor set $\mathcal{N}_i$, coordination interval $\tau$}
	\KwOut{Locally consistent actions} \vspace{0.35em}
	
	\ForEach{neighbor $j \in \mathcal{N}_i$}{
		Exchange summarized state and intent information $(\hat{S}_i, \hat{M}_i)$\;
	}\vspace{0.35em}
	
	Detect coordination conflicts (e.g., overload, redundant replication, contention)\;\vspace{0.35em}
	\eIf{conflict detected}{
		Negotiate local action adjustments minimizing joint marginal cost $\Delta \mathcal{O}_{i,j}$\;
		Apply agreed-upon action updates\;
	}{
		Maintain current policy\;
	}
	
	Wait for $\tau$ before next coordination round\;\vspace{0.35em}
\end{algorithm}

\begin{algorithm}[t]
	\caption{Execution Agent Task Handling}
	\label{alg:execution} \vspace{0.35em}
	\KwIn{Task specification from fog agent}
	\KwOut{Execution status} \vspace{0.35em}
	
	Receive task request\;\vspace{0.35em}
	Execute specified operation (serve, replicate, migrate)\;\vspace{0.35em}
	Report status (success, failure) to requesting fog agent\;\vspace{0.35em}
\end{algorithm}

The essence of agentic behavior of FNs is formalized in Algorithm \ref{alg:fog_agent}. Each FA operates under partial observability independently by combining local state and shared historical context. The algorithm explicitly considers the marginal contribution of candidate actions to the global objective which results in local decisions aligned with the global objective. The bounded-rational best-response update can make sure that the global potential function is monotonically reduced. Finally, publishing summarized outcomes to shared memory allows for the collective learning and provides for the convergence under asynchronous updates.

Algorithm \ref{alg:coordination} allows for a quick local stabilization by direct p2p coordination of FAs. By exchanging summaries of their high-level state, agents settle conflicts such as redundant replication or contention due to independent decision-making without global synchronisation. This way oscillatory behaviour is avoided and consistency is enforced. The structured interactions defined in this protocol are important for ensuring the existence of a global potential function and ensuring stability.

Algorithm \ref{alg:execution} outlines the behaviour of stateless and non-strategic EAs in charge of executing allotted tasks as well as reporting results. The EA receives a task, executes it, and reports the outcomes of the execution. This imposed decoupling allows for horizontal scalability, and the resulting stability of agentic coordination from failure at the execution level.\\

\begin{proposition}[Dynamic Goal Decomposition]
	The global objective $\mathcal{O}^{*}$ can be dynamically decomposed into sub-objectives $\{o_1,\dots,o_k\}$ such that each sub-goal corresponds to a locally creditable marginal contribution and is solvable under partial observability.
\end{proposition}

\begin{proof}
	Latency, cost, and risk decompose additively across nodes and paths: $L=\sum_i L_i$, $C=\sum_i C_i$, and $R=\sum_i R_i$. Each component depends only on local decisions and neighboring interactions. Thus, each agent can optimize a local sub-goal aligned with its marginal impact on $\mathcal{O}^{*}$.
\end{proof}
\begin{lemma}[Local Rationality]
	Given $(\zeta_{i}, \mathcal{M}_i, \mathcal{M}^{*})$, agent $\mathcal{A}_i$ can construct a consistent estimator of its marginal contribution to $\mathcal{O}^{*}$.
\end{lemma}

\begin{proof}
	Shared memory provides empirical correlations between past actions and observed system outcomes. Under bounded non-stationarity and finite-memory assumptions, conditioning on $(\zeta_{i}, \mathcal{M}_i, \mathcal{M}^{*})$ yields a consistent estimator of the expected marginal effect of agent $\mathcal{A}_i$’s actions on $\mathcal{O}^{*}$.
\end{proof}

\begin{assumption}
	Agents update policies using asynchronous, bounded-rational best-response dynamics with finite memory.
\end{assumption}

\begin{definition}[AF Game]
	The AF system induces a strategic game where each agent $\mathcal{A}_i$ minimizes a local utility representing its expected marginal contribution to $\mathcal{O}^{*}$. Eq. (\ref{localutility}) depicts the local utility of an Agent $i$:
	\begin{equation}
		\upsilon_i = \mathbb{E}[\Delta \mathcal{O}_i],
		\label{localutility}
	\end{equation}
	
\end{definition}
\begin{theorem}[Existence of Potential Function]
	The induced AF game admits an exact potential function.
\end{theorem}

\begin{proof}
	Define the potential function as $\Phi = \mathcal{O}^{*}$. By construction, unilateral action changes by any agent $A_i$ result in equivalent changes to both $\upsilon_i$ and $\Phi$, satisfying the definition of an exact PG.
\end{proof}
\begin{theorem}[Convergence]
	Under asynchronous best-response dynamics, the AF system converges to a Nash equilibrium corresponding to a locally optimal (stable) configuration.
\end{theorem}

\begin{proof}
	Exact PGs converge under asynchronous best-response updates. Bounded rationality ensures convergence to a local minimum of $\Phi$.
\end{proof}

\begin{theorem}[Stability Under Failures]
	If a subset of agents fails permanently, the system remains stable provided shared memory persists, and the communication graph remains connected.
\end{theorem}

\begin{proof}
	Agent failures reduce the feasible action space but do not alter the potential function $\Phi$. Remaining agents re-optimize their local utilities to yield a new equilibrium with bounded degradation in $\mathcal{O}^{*}$.
\end{proof}

\begin{table}[t]
	\centering
	\small
	\caption{Simulation Parameters}
	\label{tab:sim_config}
	\begin{tabular}{ll}
		\hline
		Parameter & Value \\
		\hline
		No. of Fog Nodes & 10--50 \\
		Topology & Bounded-degree random mesh \\
		Request Arrival & Non-stationary Poisson \\
		Content Popularity & Zipf ($s > 1$) \\
		Shared Memory Size & 20--100 episodes \\
		Coordination Interval & 5--25 time units \\
		Simulation Runs & $N \geq 10$ per setting \\
		\hline
	\end{tabular}
\end{table}

\section{Simulation Results and Discussion}

To validate the theoretical claims, we evaluate the proposed AF framework using extensive simulation that model dynamic, partially observable fog environments. 
FNs are deployed over a mesh network represented as an undirected graph with 10 to 50 nodes. Node degrees follow a bounded random graph with maximum degree constraints to emulate realistic fog connectivity and prevent hub dominance. 
User requests arrive according to a non-stationary Poisson process with a time-varying rate $\lambda(t)$, capturing low, medium, and high demand patterns. Content popularity follows a Zipf distribution with a skew parameter $s > 1$, inducing spatial and temporal demand imbalance across FNs. Table~\ref{tab:sim_config} summarizes the key simulation parameters used across all experiments. Results are averaged over multiple independent runs to ensure statistical significance.

Each FA updates its policy using asynchronous bounded-rational best-response dynamics conditioned on $(\zeta_{i}, \mathcal{M}_i, \mathcal{M}^{*})$. For each candidate action, agents estimate the expected marginal contribution using Eq. (\ref{agentcontribution}):
\begin{equation}
	\upsilon_i(a) = \mathbb{E}[\Delta \mathcal{O}_i \mid \zeta_{i}, \mathcal{M}_i, \mathcal{M}^{*}, a],
	\label{agentcontribution}
\end{equation}
and select actions minimizing $\upsilon_i(a)$. Policies are initialized randomly and updated without global synchronization.

We compare the proposed AF framework against the following baselines:
\\\textbf{Centralized ILP}: Global optimization solved periodically using complete system state snapshots.
\\\textbf{Greedy Heuristic}: Local latency minimization without coordination or shared memory.

Furthermore, we evaluate the performance of the proposed system against baselines using the following metrics:
\\\textbf{Average Latency:} Mean end-to-end request response time.
\\\textbf{Adaptation Time:} Time required to stabilize after workload shifts.
\\\textbf{Failure Resilience:} Relative latency degradation under random node failures.
\\\textbf{Control Overhead:} Coordination and optimization of message cost.

\subsection{Latency and Adaptivity}
Fig. \ref{fig:fig1} shows a comparative analysis of average latency under dynamic workloads for proposed AF framework, Greedy heuristic and ILP based control scheme. Proposed AF system has an improved latency reduction of 15-30\% and 10-18\% as compared to Greedy strategy and ILP approach respectively. These results confirm the efficacy of the underlying PG formulation which makes it possible to optimize trajectory (trajectory-level optimization and adaptability in highly-dynamic environments).

\begin{figure}[t]
	\centering
	\includegraphics[width=\linewidth]{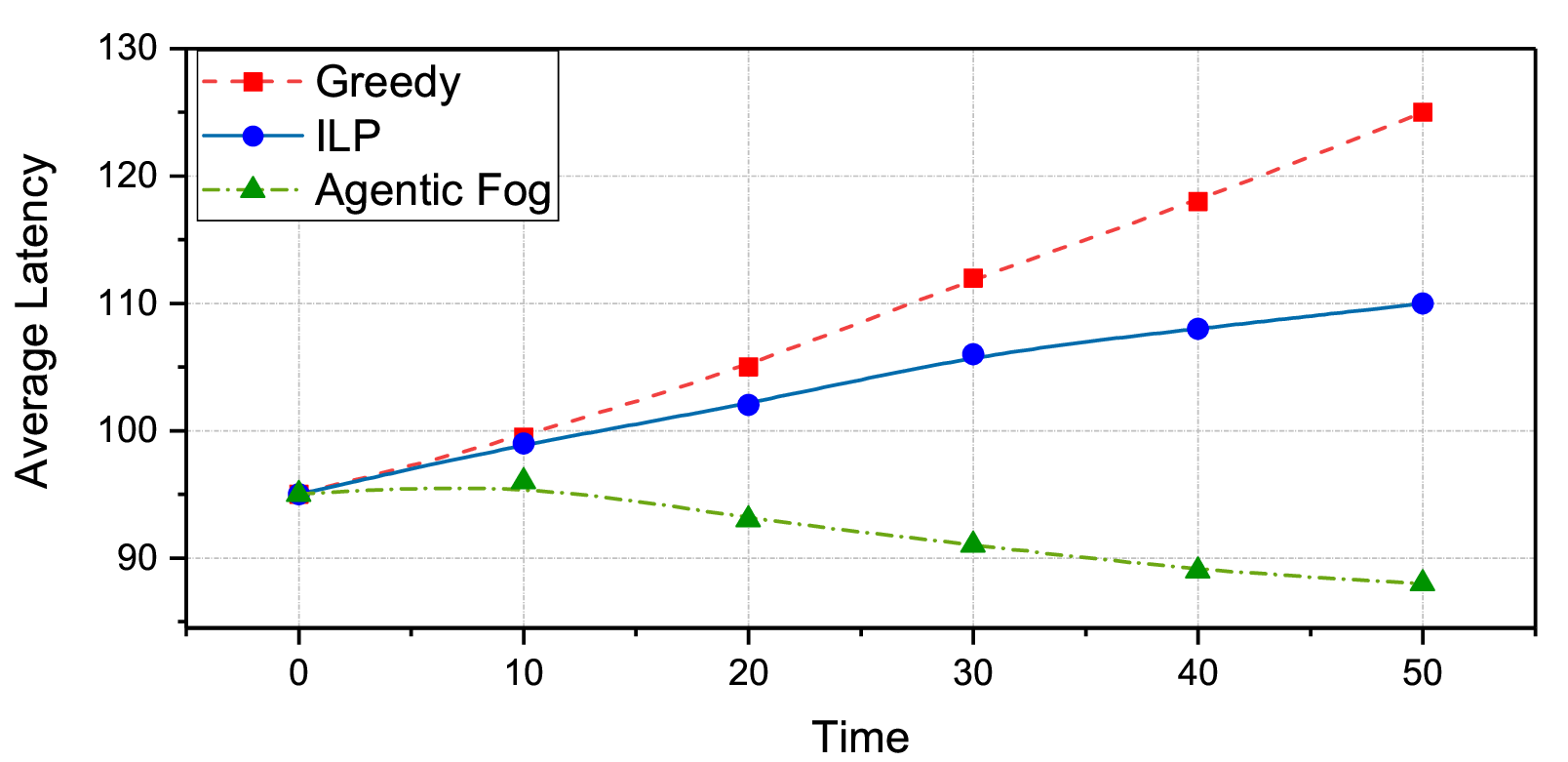}
	\caption{Comparing the latency under varying workload conditions. The proposed AF records lower average latency as compared to the baseline strategies over time.}
	\label{fig:fig1}
\end{figure}

\subsection{Convergence Behavior}
Fig. \ref{fig:fig2} compares the convergence behavior of the proposed AF system with ILP based and Greedy approaches. The AF system converges significantly faster in fewer iterations for various sizes of network, proving the existence of scalability and efficiency. Notably, the number of iterations is almost constant in line with the convergence guarantees of exact policy gradient (PG) methods under asynchronous best response dynamics. This bound convergence points to the system's suitability for large-scale, dynamic fog environments.

\begin{figure}[ht]
	\centering
	\includegraphics[width=\linewidth]{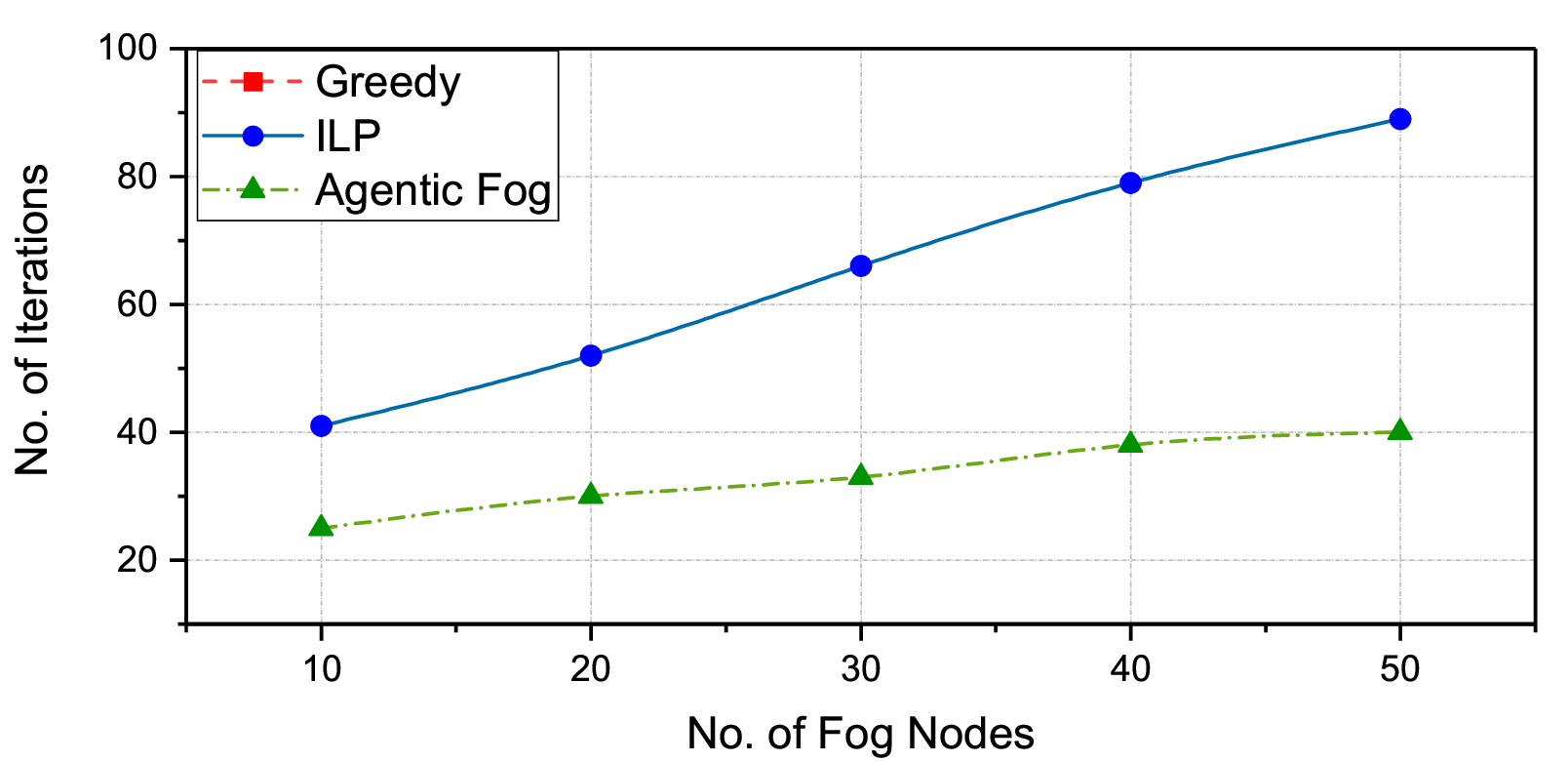}
	\caption{Comparing the convergence behavior as a function of network size. The proposed AF system results in faster convergence with increasing number of fog nodes, compared with the ILP and Greedy baselines.}
	\label{fig:fig2}
\end{figure}

\subsection{Failure Resilience}

The proposed AF system shows better stability in random fog node failure compared to ILP-based and Greedy control strategies. While the baseline methods incur high performance degradation with this being reflected in a steep increase in latency, the AF framework retains robust performance with marginal latency growth as the failure rate increases. The observed resilience could be explained by the adaptive policy mechanisms and decentralised control that the AF architecture provides and which enable the dynamic reconfiguration in response to unavailability of node, as shown in Fig. \ref{fig:fig3}.

\begin{figure}[t]
	\centering
	\includegraphics[width=\linewidth]{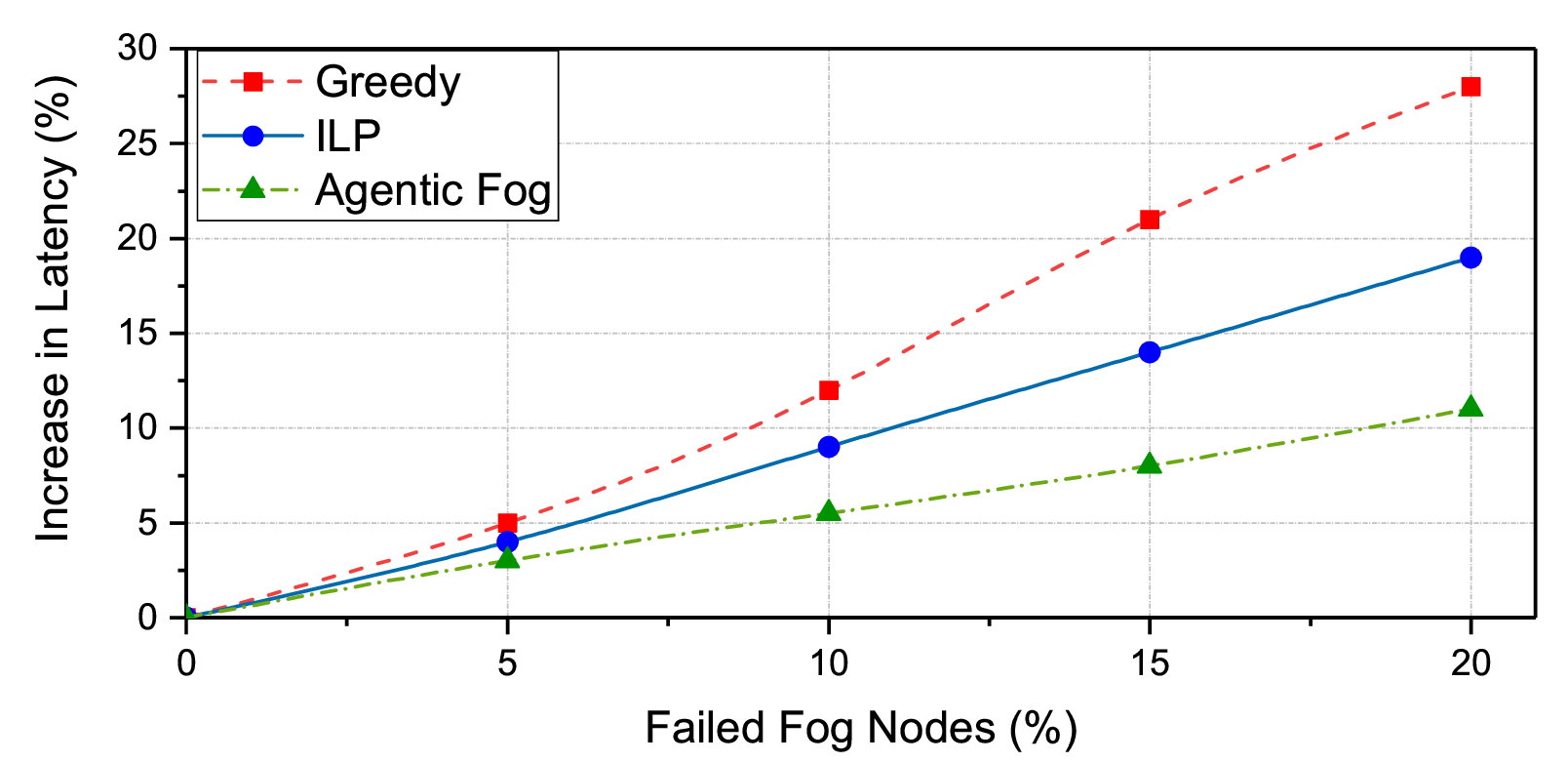}
	\caption{Performance degradation under node failures. The proposed Agentic Fog framework exhibits greater resilience to random fog node failures compared to ILP and Greedy baselines, with minimal increase in latency.}
	\label{fig:fig3}
\end{figure}

\subsection{Control Overhead}
Fig. \ref{fig:fig4} shows that the proposed approach results in slightly higher overhead in terms of number of messages passed between peer agents; however, the difference in the performance of the proposed AF and the ILP approach is negligible compared to the improvements in terms of latency and stability.

\begin{figure}[ht]
	\centering
	\includegraphics[width=\linewidth]{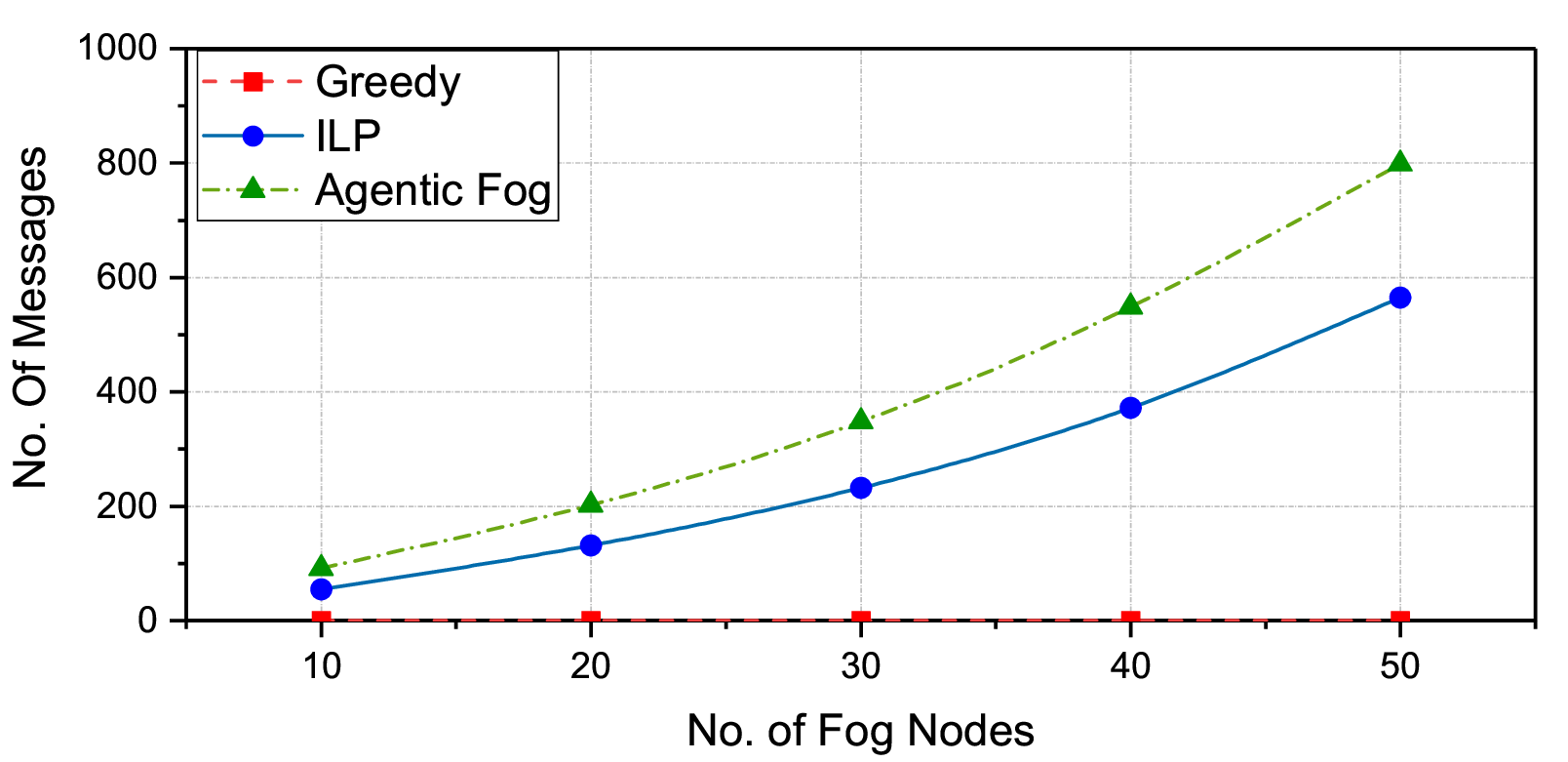}
	\caption{Coordination overhead relative to the number of fog nodes (FNs). The AF system incurs slightly more message exchange than the ILP and Greedy baselines.}
	\label{fig:fig4}
\end{figure}

\subsection{Sensitivity Analysis}
To evaluate the influence of historical context, we vary the shared memory size between 20 and 100 episodes. As shown in Fig. \ref{fig:fig5}, the performance of the AF system improves steadily with increased memory, particularly in latency reduction. However, the benefit plateaus beyond 100 episodes, indicating that additional memory yields diminishing returns in this setting.

\begin{figure}[ht]
	\centering
	\includegraphics[width=\linewidth]{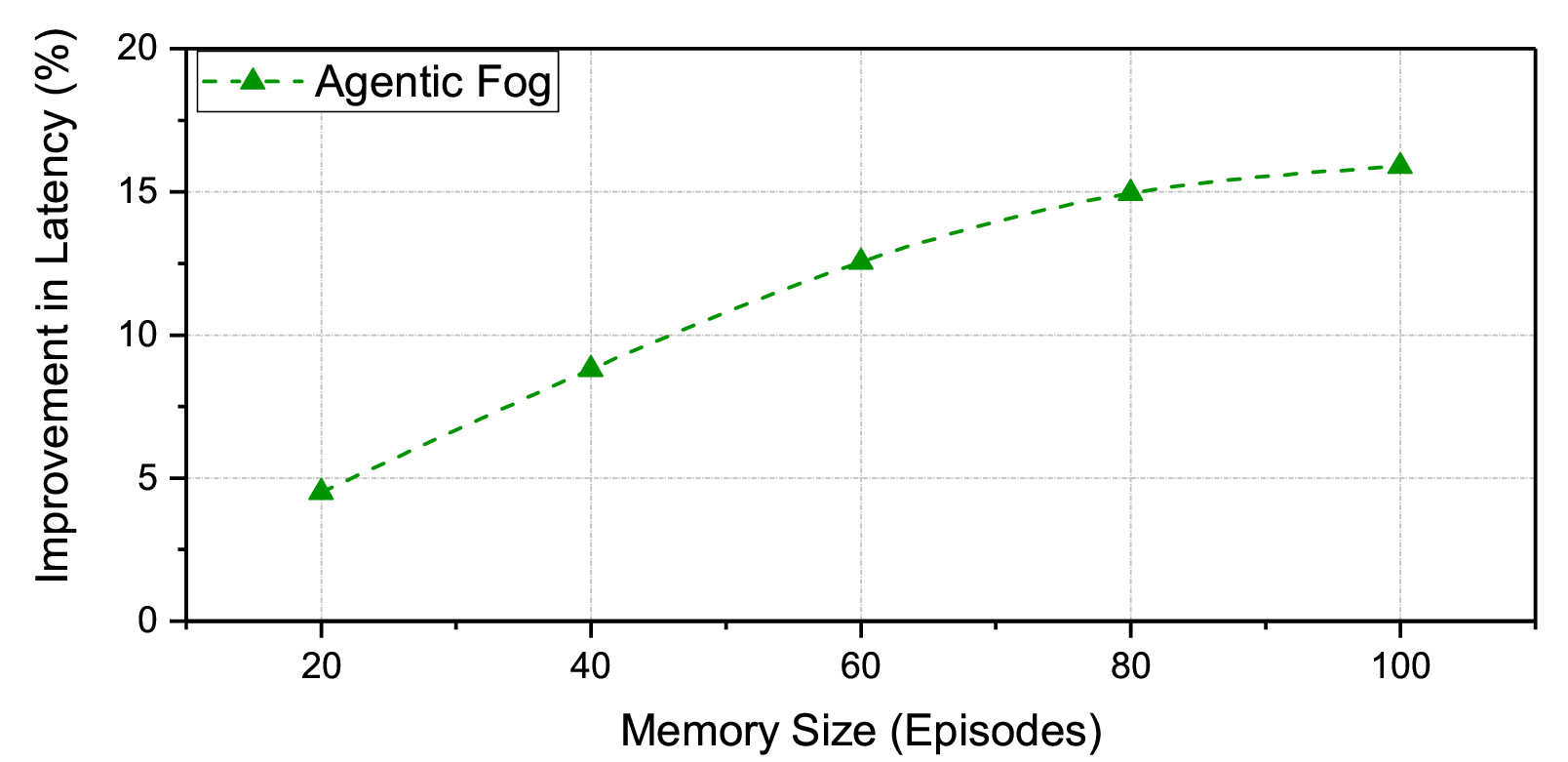}
	\caption{Impact of shared memory size on the performance of Agentic Fog. Larger memory windows yield improved latency, with diminishing returns beyond 100 episodes.}
	\label{fig:fig5}
\end{figure}

\section{Conclusion}
This paper established a rigorous theoretical and system-level foundation for AAI in fog computing environments. We model FNs as autonomous, policy-driven agents operating under partial observability and shared memory. We demonstrate that system-level intelligence emerges from localized decision-making that operates independently of central optimization and language-processing systems.

The proposed AF framework presents the formal specification of a coordination as an exact PG. This offers convergence guarantees in the case of asynchronous bounded-rational updates and ensures stability in the case of node failures. In contrast to classical fog control based on ILP, which has the assumption that everything can be observed, the aim of the AF is to optimize the system behavior over time. It also makes it possible to remain robust under varying demand and dynamic network conditions.

In addition to the performance gains, this work adds a conceptual contribution since it clarifies the difference between AAI and LLM-centric AAI frameworks. Our results demonstrate that there are agentic behaviors which are appropriate for infrastructure systems based on structured autonomy, a shared context, and coordinated policies. This clarification is especially applicable, however, to latency-sensitive, resource-constrained environments for which formal analyzability and predictable execution are necessary.

Simulation statistics confirm the theory and show better adaptivity, minimal coordination overhead, and graceful degradation in the case of partial system failures. Sensitivity analysis is further used in order to underline the trade-offs between memory, coordination frequency, and system stability. This also helps in reinforcing the practical relevance of the proposed AF system. 

In short this research reviews fog infrastructures as cognitive distributed systems and opens a principled ground for designing non-LLM agentic architectures in next generation edge and fog computing. 

{
	\renewcommand{\markboth}[2]{}
	\printbibliography
}

\end{document}